\documentclass[journal,twoside,web]{ieeecolor}
\usepackage{lcsys}
\usepackage{cite}
\usepackage{amsmath,amssymb,amsfonts}
\usepackage{graphicx}
\usepackage{textcomp}
\newtheorem{proposition}{Proposition}
\usepackage{algorithm}	
\usepackage{algpseudocode}
\usepackage{subfig}	
\usepackage{hyperref}	

\DeclareMathOperator*{\argmax}{arg\,max}

\usepackage{tikz}
\usetikzlibrary{spy}
\usepackage[utf8]{inputenc}
\usepackage{pgfplots}
\usepgfplotslibrary{groupplots,dateplot}
\usetikzlibrary{patterns,shapes.arrows}
\pgfplotsset{compat=newest}
\def\axisdefaultheight{110pt}

\def\BibTeX{{\rm B\kern-.05em{\sc i\kern-.025em b}\kern-.08em
    T\kern-.1667em\lower.7ex\hbox{E}\kern-.125emX}}
\begin{document}
\title{Learning-Based MPC for Fuel Efficient Control of Autonomous Vehicles with Discrete Gear Selection}
\author{Samuel Mallick, Gianpietro Battocletti, Qizhang Dong, Azita Dabiri, Bart De Schutter
\thanks{
	This project has received funding from the European Research Council (ERC) under the European Union's Horizon 2020 research and innovation programme (Grant agreement No. 101018826 - ERC Advanced Grant CLariNet).}
\thanks{
	The authors are with the Delft Center for Systems and Control, Delft University of Technology,
	Delft, The Netherlands,
	{\tt\small \{s.h.mallick, g.battocletti, q.dong, a.dabiri, b.deschutter\}@tudelft.nl}}%
}

\maketitle

\begin{abstract}
Co-optimization of both vehicle speed and gear position via model predictive control (MPC) has been shown to offer benefits for fuel-efficient autonomous driving.
However, optimizing both the vehicle's continuous dynamics and discrete gear positions may be too computationally intensive for a real-time implementation.
This work proposes a learning-based MPC scheme to address this issue.
A policy is trained to select and fix the gear positions across the prediction horizon of the MPC controller, leaving a significantly simpler continuous optimization problem to be solved online.
In simulation, the proposed approach is shown to have a significantly lower computation burden and a comparable performance, with respect to pure MPC-based co-optimization.
\end{abstract}


\section{INTRODUCTION}
\IEEEPARstart{F}or optimal control of autonomous vehicles (AVs), model predictive control (MPC) is a powerful and prevalent method \cite{mallick2024comparison,zhengDistributedModelPredictive2017}.
In this context, co-optimization of an AV's speed and gear-shift schedule is a promising approach to achieve high-performing and fuel-efficient autonomous driving; however, considering the gear-shift schedule requires the online solution of a mixed-integer nonlinear program (MINLP) \cite{bemporadControlSystemsIntegrating1999}, for which the computational burden can be intensive.

To address this issue the MINLP can be made easier to solve by relaxing the problem or by finding heuristic numerical solutions \cite{shaoVehicleSpeedGear2021, ganesanNumericalStrategiesMixedInteger2023}.
However, the relaxed problem can still be difficult to solve, and approximate solutions can be suboptimal.
Alternatively, a decoupled approach can be used.
In this case, first the speed is optimized, and next a gear-shift schedule is selected for the given speed using, e.g., a learning-based gear controller \cite{yinHierarchicalModelPredictive2022}, or dynamic programming \cite{turriGearManagementFuelEfficient2016}.
However, decoupling speed control and gear control is suboptimal compared to the joint optimal control of both together.
Finally, the computational burden can be alleviated by using a learning-based controller that controls both speed and gear-shift schedule in place of an MPC controller \cite{liEcologicalAdaptiveCruise2020}; however, in contrast to an MPC controller, a learning-based controller is not able to guarantee constraint satisfaction.

In light of the above issues, this work presents a novel learning-based MPC controller for the co-optimization of speed and gear-shift schedule for an AV.
Taking inspiration from \cite{da2024integrating}, a learned policy selects and fixes the gear positions across the prediction horizon of the MPC controller, such that optimal control and constraint satisfaction are handled by a nonlinear program (NLP), rather than an MINLP. 
A neural network (NN)-based policy is proposed where, to address the exponential growth of the policy's action space with the prediction horizon, a recurrent architecture is used.
The policy learns to select gears that are optimal for the original optimization problem, rather than decoupling the gear-shift schedule from the speed control, such that, in contrast to \cite{yinHierarchicalModelPredictive2022, turriGearManagementFuelEfficient2016}, the notion of co-optimization is retained. 
In this way, the MPC controller is able to consider the gear and powertrain dynamics without optimizing explicitly over discrete inputs, in contrast to \cite{shaoVehicleSpeedGear2021, ganesanNumericalStrategiesMixedInteger2023}.
Furthermore, unlike pure learning-based controllers \cite{liEcologicalAdaptiveCruise2020}, the use of the MPC controller gives constraint satisfaction.
To this end, we propose a backup gear-shift schedule that can guarantee feasibility of the MPC optimization problem.
Finally, due to the recurrent architecture, the policy, once trained for a specific prediction horizon, generalizes over prediction horizons without the need for retraining.

\section{PROBLEM SETTING}\label{sec:problem_setting}
Consider the vehicle and powertrain models \cite{shaoVehicleSpeedGear2021}
\begin{equation}\label{eq:dynamics}
	\begin{aligned}
		T(t)n(t) &= g(t) + C v^2(t) + m a(t) + F_\text{b}(t), \\
		\omega(t) &= (30/\pi) \cdot n(t) v(t) ,
	\end{aligned}
\end{equation} 
with $t$ continuous time, $a$ the acceleration, $v$ the velocity, and $m$ the mass of the vehicle.
Furthermore, $C$ is the wind drag coefficient, $F_\text{b}$ is the brake force, $T$ is the engine torque, and $\omega$ is the engine speed.
The friction function 
\begin{equation}
	g(t) = \mu m g \cos\big(\alpha(t)\big) + m g \sin\big(\alpha(t)\big), 
\end{equation}
with $\mu$ the rolling friction constant and $g$ the gravitational acceleration, defines the road friction for road angle $\alpha$, which, for simplicity of presentation, is assumed to be constant, i.e., $\alpha(t) = \alpha$ and $g(t) = G$.
The lumped gear ratio $n(t) = z\big(j(t)\big)z_\text{f}/r$ is determined by the final drive ratio $z_\text{f}$, the wheel radius $r$, and the transmission gear ratio $z$, a discrete variable selected by the gear position $j \in \{1,\dots,6\}$.
To model engine dynamics the torque rate of change is constrained as $|\dot{T}(t)| \leq \Delta T_\text{max}.$
Furthermore, consider the fuel model \cite{shaoVehicleSpeedGear2021}
\begin{equation}\label{eq:fuel}
	\dot{m}_{\text{f}}(t) = c_0 + c_1 \omega(t) + c_2 \omega(t) T(t),
\end{equation}
with $c_{0}, c_1,$ and $c_2$ constants.
The variables $F_\text{b},T,$ and $\omega$ are physically bounded above and below, e.g., $T_\text{min} \leq T \leq T_\text{max}$.
Note that the bounds on $\omega$ implicitly bound $v$ between
\begin{equation}
	v_\text{min} = \frac{\pi \cdot \omega_\text{min} r}{30 \cdot z(1) z_\text{f}} \quad \text{and}\quad v_\text{max} = \frac{\pi \cdot \omega_\text{max} r}{30 \cdot z(6) z_\text{f}}.
\end{equation}
For convenience, in the following, we define a function that maps a speed and gear choice to the engine speed
\begin{equation}
	\omega(v, j) = \frac{30 \cdot v \cdot z(j) z_\text{f}}{r \pi}.
\end{equation}

We consider the control of an AV to track a reference trajectory in a fuel efficient manner.
Denote the vehicle position, reference position, and reference velocity at time $t$ as $p(t)$, $p_\text{ref}(t)$, and $v_\text{ref}(t)$, respectively.
The performance metric is
\begin{equation}\label{eq:metric}
	\begin{aligned}
		P &= \sum_{k = 0}^{K_\text{sim}} \beta L_\text{t}\Big(\big[p(k \Delta t), v(k \Delta t)\big]^\top, \big[p_\text{ref}(k \Delta t), v_\text{ref}(k \Delta t)\big]^\top\Big) \\ &\quad + L_\text{f}\big(v(k \Delta t), T(k \Delta t), j(k \Delta t)\big),
	\end{aligned}
\end{equation}
where $\beta > 0$ expresses the importance of tracking against fuel efficiency, and $k$ is a discrete-time counter for time steps of $\Delta t$ seconds. 
The tracking cost
\begin{equation}
	L_\text{t}(x, y) = (x - y)^\top Q (x-y)
\end{equation}
quadratically penalizes deviations from the reference trajectory, with $Q \in \mathbb{R}^{2 \times 2}$ a positive-definite weighting matrix.
The fuel cost $L_\text{f}(v, T, j) = \Delta t \big(c_0 + c_1 \omega(v, j) + c_2 \omega(v, j) T\big)$ penalizes the fuel consumption over a time step.

\section{LEARNING-BASED MPC}\label{sec:controller}
In this section we introduce the proposed controller for the task.
Defining the state, reference state, and control input as 
\begin{equation}
		\begin{aligned}
			x(k) &= [p(k \Delta t), v(k \Delta t)]^\top, \\
			x_\text{ref}(k) &= [p_\text{ref}(k \Delta t),  v_\text{ref}(k \Delta t)]^\top, \\
			u(k) &= [T(k \Delta t), F_\text{b}(k \Delta t), j(k \Delta t)]^\top,
		\end{aligned}
\end{equation}
\eqref{eq:dynamics} can be approximated with the discrete-time dynamics $x(k+1) = f\big(x(k), u(k)\big)$, with
\begin{equation}\label{eq:disc_dynam}
	f(x, u) = \begin{bmatrix}
		x_1 +  x_2 \Delta t\\
		x_2 +  \frac{\Delta t}{m}\Big(\frac{u_1 z(u_3) z_\text{f}}{r} - C x_2^2 - u_2 - G\big)
		\Big)	\end{bmatrix},
\end{equation}
where subscripts select an element of a vector, e.g., $x_i$ is the $i$th element of the vector $x$.

\subsection{Mixed-integer nonlinear MPC}
Consider an MPC scheme with prediction horizon $N > 1$ defined by the following MINLP:
\begin{subequations}\label{eq:MINLP}
	\begin{align}
		J\big(x(k), &\textbf{x}_\text{ref}(k)\big) = \min_{\textbf{x}(k), \textbf{u}(k)}  \sum_{i = 0}^{N} \beta L_\text{t}\big(x(i|k), x_\text{ref}(i+k)\big) \nonumber \\
		&\quad \quad+ \sum_{i=0}^{N-1} L_\text{f}\big(x_2(i|k), u_1(i|k), u_3(i|k)\big) \\
		\text{s.t.} \quad &x(0|k) = x(k) \label{eq:IC}\\
		&\text{for} \quad i = 0,\dots,N-1: \nonumber \\
		\vspace{0.5cm}
		&\quad x(i+1|k) = f\big(x(i|k), u(i|k)\big) \label{eq:dynam_cnstr} \\
		&\quad |x_2(i+1|k) - x_2(i|k)| \leq a_\text{max}\Delta t	\label{eq:acc_lim}\\ 
		&\text{for} \quad i = 0,\dots,N-2: \nonumber \\
		&\quad |u_1(i+1|k) - u_1(i|k)| \leq \Delta T_\text{max} \Delta t \label{eq:T_lim} \\
		&\quad |u_3(i+1|k) - u_3(i|k)| \leq 1  \label{eq:gear_shift_lim} \\
		&\big(\textbf{x}(k), \textbf{u}(k)\big) \in \mathcal{C} \label{eq:state_action_set_cnstr}
	\end{align}
\end{subequations}
where $x(i|k)$ and $u(i|k)$ are the predicted states and inputs, respectively, $i$ steps into the prediction horizon of the MPC controller at time step $k$.
Furthermore, bold variables gather a variable over the prediction horizon, e.g., 
\begin{equation}
	\begin{aligned}
		\textbf{x}(k) &= \big[x^\top(0|k),\dots,x^\top(N|k)\big]^\top, \\
		\textbf{x}_\text{ref}(k) &= \big[x_\text{ref}^\top(k),\dots,x_\text{ref}^\top(k+N)\big]^\top.
	\end{aligned}
\end{equation}
If no solution exists for \eqref{eq:MINLP} $J\big(x(k), \textbf{x}_\text{ref}(k)\big) = \infty$ by convention.
The constraint \eqref{eq:T_lim} enforces the engine torque dynamic behavior, \eqref{eq:gear_shift_lim} prevents skipping gears when shifting, and \eqref{eq:acc_lim} limits acceleration to $a_\text{max}$.
The bounds on engine torque, engine speed, and brake force are grouped in
\begin{equation}
	\begin{aligned}
		\mathcal{C} = \bigg\{(\textbf{x}, \textbf{u}) \Big| &T_\text{min} \leq u_1(i|k) \leq T_\text{max}, \\
		&F_\text{b, min} \leq u_2(i|k) \leq F_\text{b, max}, \\
		&w_\text{min} \leq \omega\big(x_2(i|k), u_3(i|k)\big) \leq w_\text{max}, \\
		&w_\text{min} \leq \omega\big(x_2(i+1|k), u_3(i|k)\big) \leq w_\text{max}, \\
		&i = 0,\dots,N-1\bigg\}.
	\end{aligned}
\end{equation}
The last condition in $\mathcal{C}$, relating $x_2(i+1|k)$ to $u_3(i|k)$, ensures that the gear at time step $k+i$ maintains the engine speed within its bounds for all $t \in \big[(k+i) \Delta t, (k+i+1)\Delta t\big]$.

The MINLP \eqref{eq:MINLP} provides a state feedback controller via solving \eqref{eq:MINLP} at each time step $k$ and applying the first element $u^\ast(0|k)$ of the optimal sequence $\textbf{u}^\ast(k)$ to the system.
However, the computation required to solve \eqref{eq:MINLP} online renders it unsuitable for a real-time implementation, \begin{color}{black}i.e., the computation time is larger than the available  real time for the MPC to make a decision.\end{color} 
In the following, we introduce an alternative MPC controller that can be executed efficiently online.

\subsection{Learning-based nonlinear MPC}
Let us define a reduced control action that does not include the gear choice as $u^\prime(k) = \big[T(k \Delta t), F_\text{b}(k \Delta t)\big]^\top$.
We then introduce the following MPC controller, parameterized by a gear-shift schedule $\textbf{j}(k) = \big[j(0|k),\dots,j(N-1|k)\big]^\top$:
\begin{subequations}\label{eq:NLP}
	\begin{align}
		J&\big(x(k), \textbf{x}_\text{ref}(k), \textbf{j}(k)\big) = \min_{\substack{\textbf{x}(k),\\ \textbf{u}^\prime(k)}}  \sum_{i = 0}^{N} \beta L_\text{t}\big(x(i|k), x_\text{ref}(i + k)\big) \nonumber \\
		&\quad + \sum_{i=0}^{N-1} L_\text{f}\big(x_2(i|k), u_1^{\prime}(i|k), j(i|k)\big) \\
		\text{s.t.} \quad &\eqref{eq:IC}, \eqref{eq:acc_lim}, \eqref{eq:T_lim} \label{eq:nlp_rate_cnstr}\\
		\vspace{0.1cm}
		&\big(\textbf{x}(k), [u^{\prime, \top}(0|k), j(0|k),\dots, \nonumber \\
		&\quad\quad u^{\prime, \top}(N-1|k), j(N-1|k)]^\top\big) \in \mathcal{C} \label{eq:nlp_state_input_cnstr}\\
		&\text{for }i = 0,\dots,N-1 \nonumber \\ 
		&\quad x(i+1|k) = f\big(x(i|k), [u^{\prime, \top}(i|k), j(i|k)]^\top\big) \label{eq:nlp_dynam_cnstr} \\
		&\textcolor{black}{\text{for }i = 0,\dots,N-2:} \nonumber \\ 
		&\quad \textcolor{black}{|j(i+1|k) - j(i|k)| \leq 1}  \label{eq:nlp_gear_shift_lim}.
	\end{align}
\end{subequations}
With $\textbf{j}$ prespecified, no discrete variables are optimized in the problem \eqref{eq:NLP}, which can now be solved efficiently using numerical nonlinear solvers. 
Note that if $\textbf{j} = \textbf{u}_3^\ast$, the optimal gear-shift schedule from \eqref{eq:MINLP}, then $J(x, \textbf{x}_\text{ref}, \textbf{j}) = J(x, \textbf{x}_\text{ref})$.

We propose the use of a learned policy that selects and fixes the gears over the prediction horizon based on the optimal solution to the MPC problem from the previous time step.
Define the shifted solutions to \eqref{eq:NLP} at time step $k$, i.e., the optimal control and state trajectories from time step $k-1$ advanced by one time step, as 
	\begin{equation*}
		\begin{aligned}
			&\bar{\textbf{x}}(k) = \big[x^\top(k), x^{\ast, \top}(2|k-1), \scriptsize{\dots}, x^{\ast, \top}(N|k-1)\big]^\top \hspace{-0.15cm} \in \mathbb{R}^{2N}, \\
			&\bar{\textbf{u}}^\prime(k) = \big[u^{\prime, \ast, \top}(1|k-1),\dots, \\
			&\quad u^{\prime, \ast, \top}(N-1|k-1), u^{\prime, \ast, \top}(N-1|k-1)\big]^\top \in \mathbb{R}^{2N}.
		\end{aligned}
	\end{equation*}
	Note that the first element of $\bar{\textbf{x}}(k)$ is replaced with the state $x(k)$, such that in the case of modeling errors the real state is present.
Furthermore, define the shifted gear-shift schedule
\begin{equation}
	\bar{\textbf{j}}(k) = \big[j(1|k-1), \dots, j(N-1|k-1), j(N-1|k-1)\big]^\top.
\end{equation}
Consider the selection of $\textbf{j}$ by a policy, parameterized by $\theta$,
\begin{equation}
	\textbf{j} = \pi_\theta(\bar{\textbf{x}}, \bar{\textbf{u}}^\prime, \textbf{x}_\text{ref}, \bar{\textbf{j}}), 
\end{equation}
as a function of the reference trajectory and the shifted solutions from the previous time step\footnote{If $\alpha$ would be time-varying then $\alpha$ would be an additional input argument.}.
Note that for simplicity the time index $(k)$ is dropped.
In Section \ref{sec:learning} the architecture and training of $\pi_\theta$ are described.

Observe that there are many choices for $\textbf{j}$ for which \eqref{eq:NLP} has no solution.
While we can expect $\pi_\theta$ to almost always provide at least a feasible $\textbf{j}$ (if not optimal), here we propose a \emph{backup solution} that will be useful for guaranteeing feasibility at deployment.
\begin{color}{black}We then prove that this backup solution, while potentially suboptimal, always provides a feasible $\textbf{j}$.\end{color}
Define the set of feasible gears for a given velocity $v$ as:
\begin{equation}
	\Phi(v) = \Big\{j \in \{1,\dots,6\} \Big| w_\text{min} \leq \omega(v, j) \leq w_\text{max} \Big\},
\end{equation} 
\begin{color}{black}and $\phi$ as any mapping from $v$ to one of the gears $j \in \Phi(v)$, e.g., $\phi(x_2) = \max_{j^\prime \in \Phi(x_2)}j^\prime$.\end{color}
\begin{color}{black}The backup solution, defined by the function $\sigma$, is then
	\begin{equation}\label{eq:backup}
		\textbf{j} = \sigma(x_2) = \big[\phi(x_2), \dots, \phi(x_2)\big]^\top.
\end{equation}\end{color}
Furthermore, define a map from gears to velocities that satisfy the engine speed constraints as $\Omega(j) = \{v | j \in \Phi(v)\}$.
\begin{proposition}\label{prop:1}
	Assume that, for $j \in \{1,\dots,6\}$ and for all $x_2 \in \Omega(j)$, there exist $u_1$ and $u_2$ such that $T_\text{min} \leq u_1 \leq T_\text{max}$, $ F_\text{b, min} \leq u_2 \leq F_\text{b, max}$, and
	\begin{equation}\label{eq:prop1_assumption}
		\begin{aligned}
			\frac{u_1 z(j) z_\text{f}}{r} - C x_2^2 - u_2 - G = 0.
		\end{aligned}
	\end{equation}
	Then, for a state $x(k)$ such that $v_\mathrm{min} \leq x_2(k) \leq v_\mathrm{max}$, and a gear-shift sequence $\textbf{j}(k) = \sigma\big(x_2(k)\big)$, problem \eqref{eq:NLP} has a solution, i.e., $J\big(x(k), \textbf{x}_\mathrm{ref}(k), \textbf{j}(k)\big) < \infty$.
\end{proposition}
\begin{proof}
	See Appendix \ref{proof:prop_1}.
\end{proof}
As $x_2^2$ is positive and monotonic, verifying condition \eqref{eq:prop1_assumption} for a given vehicle involves checking only the endpoints of the range $\Omega(j)$ for $j=1,\dots,6$, yielding $6\times2$ conditions to be verified.
The condition is satisfied for reasonable vehicle parameters, including those used in Section \ref{sec:simulations}.
Note that Proposition \ref{prop:1} guarantees instantaneous feasibility of \eqref{eq:NLP}.
Recursive feasibility follows trivially if the true underlying system is \eqref{eq:disc_dynam}.
In the case of modeling errors, e.g., when the dynamics \eqref{eq:disc_dynam} are a discrete-time approximation of the continuous-time system \eqref{eq:dynamics}, a robust MPC formulation would be required for recursive feasibility and constraint satisfaction.
This is left for future work.

The proposed control algorithm is given in Algorithm \ref{alg:controller}.
If the gear-shift schedule proposed by $\pi_\theta$ is infeasible, the control input is computed using the backup gear-shift schedule \eqref{eq:backup}, which is feasible by Proposition \ref{prop:1}.
\begin{algorithm}
	\small
	\caption{Control algorithm at time step $k$}\label{alg:controller}
	\begin{algorithmic}
		\State \textbf{Inputs}: $x(k)$, $\textbf{x}_\text{ref}(k)$, $\bar{\textbf{x}}(k)$, $\bar{\textbf{u}}^\prime(k)$, and $\bar{\textbf{j}}(k)$
		\State $\textbf{j}(k) \gets \pi_\theta\big(\bar{\textbf{x}}(k), \bar{\textbf{u}}^\prime(k), \textbf{x}_\text{ref}(k), \bar{\textbf{j}}(k)\big)$
		\State Solve \eqref{eq:NLP} for $J\big(x(k), \textbf{x}_\text{ref}(k), \textbf{j}(k)\big)$ and $\textbf{u}^{\prime, \ast}( k), \textbf{x}^\ast(k)$
		\State\quad\textbf{If} $J=\infty$ \textbf{then} $\textbf{j}(k) \gets \sigma\big(x_2(k)\big)$
		\State \quad\quad Solve \eqref{eq:NLP} for $J\big(x(k), \textbf{x}_\text{ref}(k), \textbf{j}(k)\big)$ and $\textbf{u}^{\prime, \ast}( k), \textbf{x}^\ast(k)$
		\State Apply $\big[u^{\prime, \ast, \top}(0 | k), j(0 | k)\big]^\top$ to the system
	\end{algorithmic}
\end{algorithm}

\section{GEAR-SHIFT SCHEDULE POLICY}\label{sec:learning}
For convenience in the following, define $q(i)$ to stack the $i$th elements from $\bar{\textbf{x}}, \bar{\textbf{u}}^\prime, \textbf{x}_\text{ref}$ and $\bar{\textbf{j}}$, e.g., 
	\begin{equation}
		\begin{aligned}
			&q(0) = \big[x^\top(k), u^{\prime, \ast, \top}(1|k-1), x_\text{ref}^\top(k), j(1|k-1)\big]^\top \\
			&q(N-1) = \big[x^{\ast, \top}(N|k-1) , u^{\prime, \ast, \top}(N-1|k-1), \\
			&\quad\quad\quad\quad\quad\quad x_\text{ref}^\top(k+N-1), j(N-1|k-1)\big]^\top.
		\end{aligned}
\end{equation}
We propose to model the policy $\pi_\theta$ with an NN, with the parameter $\theta$ the model weights.
Representing $\pi_\theta$ with a standard feed-forward NN has the key issue that the action space grows exponentially with the prediction horizon $N$, while the input space grows linearly.
Indeed, for a given $N$ there are $6^N$ possible gear-shift schedules and $4(N+1)+2N$ inputs (from $\bar{\textbf{x}},\textbf{x}_\text{ref}, \bar{\textbf{u}}^\prime,$ and $\bar{\textbf{j}}$).
An NN capable of representing the input-output mapping as $N$ increases may need to be very large and highly complex.
Furthermore, there is an explicit temporal relationship between gear-shifts, which is not structurally enforced in a feed-forward NN.
In light of these points, inspired by \cite{da2024integrating} we propose a sequence-to-sequence recurrent architecture using a recurrent NN (RNN), as shown in Figure \ref{fig:RNN}.
The inputs $\bar{\textbf{x}}, \bar{\textbf{u}}^\prime, \textbf{x}_\text{ref}$ and $\bar{\textbf{j}}$ are considered as $N$ different inputs in a chain, i.e., the vectors $q(i) \in \mathbb{R}^7$ for $i=0,\dots,N-1$.
A single RNN is trained with input space $\mathbb{R}^7$, where the output is a single gear position.
The output sequence of $N$ gear positions is then generated by sequentially evaluating the RNN on the chain of inputs $q(0),\dots, q(N-1)$.
In this way the recurrent structure results in a constant number of inputs and outputs for the network for any prediction horizons, with only the number of sequential evaluations changing with the horizon.
Furthermore, the temporal relationship is structurally enforced, i.e., the gear at time step $i+k$ considers the prior gears and inputs via the hidden state $h_i$.
\begin{figure}
	\centering
	\vspace{0.2cm}
	\input{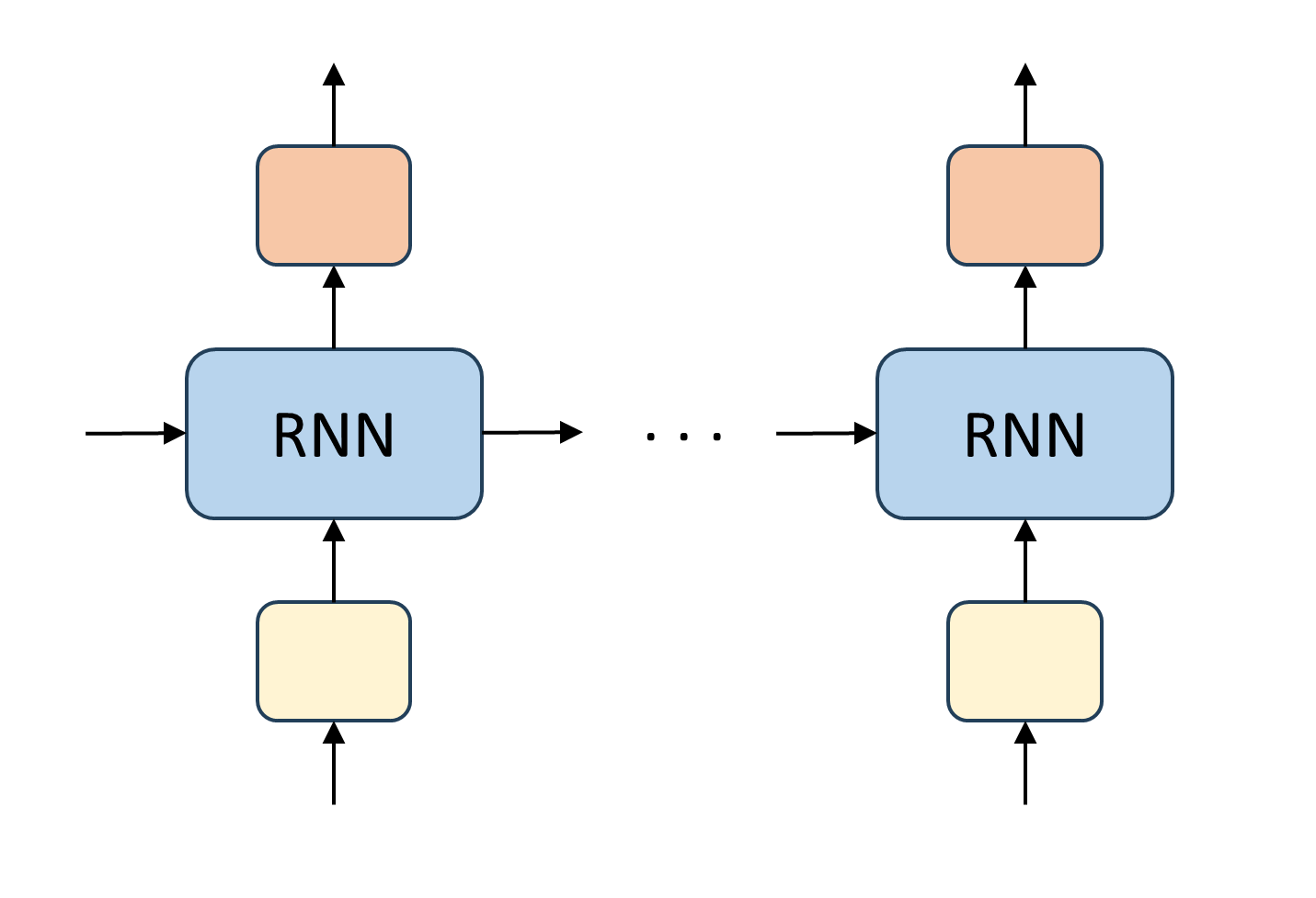}
	\caption{Recurrent NN, with hidden states $h_i$, showing how the chain of inputs sequentially generates the gear-shift schedule. The maps $\psi$ and $\eta$ are input and output transformations.}
	\label{fig:RNN}
	\vspace{-0.5cm}
\end{figure}

More formally the policy is defined as
\begin{equation}
	\begin{aligned}
		&\pi_\theta(\bar{\textbf{x}}, \bar{\textbf{u}}^\prime, \textbf{x}_\text{ref}, \bar{\textbf{j}}) = \bigg[\eta\bigg(y\Big(\psi\big(q(0)\big), h_0\Big)\bigg), \dots, \\
		&\quad \eta\bigg(y\Big(\psi\big(q(N-1)\big), h_{N-1}\Big)\bigg)\bigg]^\top,
	\end{aligned}
\end{equation}
where the input mapping $\psi$, defined by
\begin{equation}
	\begin{aligned}
		\psi\big(&q = [x^\top, u^{\prime, \top}, x_\text{ref}^\top, j]^\top\big) = \bigg[(x-x_\text{ref})^\top, \\
		&\frac{x_2 - v_\text{min}}{v_\text{max} - v_\text{min}}, \quad\frac{x_{2, \text{ref}} - v_\text{min}}{v_\text{max} - v_\text{min}}, u^{\prime, \top}, \omega(x_2, j), j\bigg]^\top,
	\end{aligned}
\end{equation}
transforms the inputs into a representation that contains the tracking error, the vehicle and reference velocities, and the predicted inputs, including the engine speed.
This representation is chosen to give the RNN the most relevant information for selecting $\textbf{j}$.
The model function of the RNN is
\begin{equation}
	\delta(i) = \big[\delta_1(i), \dots, \delta_6(i)\big]^\top = y\Big(\psi\big(q(i)\big), h_i\Big),
\end{equation} 
where $\delta_j(i)$ is the probability of choosing gear $j$ at the $i$th output in the sequence.
The output mapping $\eta$ selects the gear with the largest probability and applies a clipping: $\eta\big(\delta(0)\big) = \lambda_0$, and for $i > 0$
	\begin{equation}
		\eta\big(\delta(i)\big) = \begin{cases}
			\lambda_i & \text{if }-1 \leq \lambda_i - j(i-1|k) \leq 1 \\
			j(i-1|k)-1 & \text{if }\lambda_i - j(i-1|k) < -1 \\
			j(i-1|k)+1 & \text{if }\lambda_i - j(i-1|k) > 1
		\end{cases},
	\end{equation}
	with $\lambda_i = \argmax_j\delta_j(i)$.
	The clipping ensures that $\textbf{j}$ can be feasible by satisfying constraint \eqref{eq:nlp_gear_shift_lim}.

The policy $\pi_\theta$ can be trained in a supervised manner using a dataset of input-output pairs
\begin{equation}
	\mathcal{T} = \big\{\big((\bar{\textbf{x}}_l, \bar{\textbf{u}}^\prime_l, \textbf{x}_{\text{ref}, l}, \bar{\textbf{j}}_l), \textbf{j}_l\big)\big\}_{l=1}^{N_\text{data}}.
\end{equation}
In Section \ref{sec:simulations} the collection of $\mathcal{T}$ is detailed.
With the number of inputs and outputs of the model function $y$ independent of the prediction horizon $N$, an added benefit of the RNN architecture is that, once trained, the policy can be applied to an MPC controller with larger $N$ by applying longer input sequences.
Furthermore, the policy can be trained with data generated by different controllers with different horizons $N$. 

\section{COMPARISON CONTROLLERS}\label{sec:comparison_controllers}
In this section we outline three controllers against which the proposed method will be evaluated.

\textbf{MINLP-based MPC}: This MPC controller solves the MINLP \eqref{eq:MINLP} at each time step $k$, applying $u^\ast(0|k)$ to the system.
This controller provides the baseline performance for all other controllers, but is highly computationally intensive.

\textbf{Mixed-integer quadratic program-based (MIQP) MPC}: This controller follows the approach from \cite{shaoVehicleSpeedGear2021}, where the MINLP \eqref{eq:MINLP} is relaxed such that the remaining optimization problem is an MIQP.
In particular, a McCormick relaxation is applied to the bi-linear term in $L_\text{f}$, the quadratic term in the dynamics is replaced by a piecewise-linear approximation, and all bi-linear terms in the dynamics, e.g., $u_1 z(u_3)$, are replaced by mixed-integer inequalities (see \cite{shaoVehicleSpeedGear2021} for details).

\textbf{Hierarchical MPC}: This controller follows the principle of decoupling the optimization of the vehicle speed from the gear-shift schedule.
To this end, the simplified dynamics $x(k+1) = \tilde{f}\big(x(k), F(k)\big)$ are considered, where
\begin{equation}\label{eq:approx_dynam}
	\tilde{f}(x, F) = \begin{bmatrix}
		x_1 + \Delta t x_2 \\
		x_2 + \frac{\Delta t}{m}(F - C x_2^2 - G)
	\end{bmatrix}.
\end{equation}
The input $F$ replaces $T z(j)z_\text{f}/r - F_\text{b}$, the desired braking force and the applied force from the engine torque combined with the gear.
The following NLP is solved:
\begin{subequations}\label{eq:heirachical_NLP}
	\begin{align}
		J\big(x(k), &\textbf{x}_\text{ref}(k)\big) = \min_{\substack{\textbf{x}(k),\\ \textbf{F}(k)}} \sum_{i = 0}^{N} L_\text{t}\big(x(i|k), x_\text{ref}(i+k)\big) \\
		\text{s.t.} \quad &\eqref{eq:IC}, \eqref{eq:acc_lim}\\
		&\text{for} \quad i = 0,\dots,N-1: \nonumber \\
		&\quad x(i+1|k) = \tilde{f}\big(x(i|k), F(i|k)\big) \\
		&\quad T_\text{min}\frac{z(0) z_\text{f}}{r} - F_\text{b, max} \leq F(i|k) \leq F_\text{max}(k) \\
		&v_\text{min} \leq x_2(i|k) \leq v_\text{max} \quad i = 0,\dots,N 
	\end{align}
\end{subequations}
where the fuel cost cannot be considered as the powertrain dynamics are not modeled.
The bound $F_\text{max}(k)$ is determined at each time step $k$ by considering the gear that provides the most traction for the current velocity
\begin{equation}
	F_\text{max}(k) = T_\text{max} \cdot \max_{j \in \Phi(x_2(k))} \frac{z(j) z_\text{f}}{r}.
\end{equation}
The gear is then selected as $j(k\Delta t) = \phi\big(x_2(k)\big)$ and clipped such that \eqref{eq:gear_shift_lim} is respected.
In our simulations we found $j(k\Delta t) = \phi\big(x_2(k)\big) = \max_{j^\prime \in \Phi(x_2(k))}j^\prime$ to perform the best.
Finally, $T(k\Delta t)$ and $F_\text{b}(k\Delta t)$ are decided as 
\begin{equation}
	\begin{aligned}
		&T(k\Delta t) = \begin{cases}
			T_\text{min} & F^\ast(0|k) < 0 \\
			\frac{F^\ast(0|k) r}{z(j(k\Delta t)) z_\text{f}} & F^\ast(0|k) \geq 0
		\end{cases}, \\
		&F_\text{b}(k\Delta t) = \begin{cases}
			-F^\ast(0|k) + \frac{T_\text{min} z(j(k\Delta t)) z_\text{f}}{r} & F^\ast(0|k) < 0 \\
			0 & F^\ast(0|k) \geq 0
		\end{cases},
	\end{aligned}
\end{equation}
with the torque rate constraint \eqref{eq:T_lim} applied with clipping.

\section{SIMULATIONS}\label{sec:simulations}
In the following, MINLP problems are solved with Knitro \cite{byrd2006k}, MIQP problems are solved with Gurobi \cite{gurobi}, and NLP problems are solved with Ipopt \cite{wachter2006implementation}.
These solvers are each state-of-the-art for the respective type of optimization problem.
All coefficients defining the vehicle model can be found in \cite{shaoVehicleSpeedGear2021}, with bounds given in Table \ref{tab:dynamics_coefficients}.
Source code is available at \url{https://github.com/SamuelMallick/mpcrl-vehicle-gears}.
\begin{table}[h]
	\scriptsize	
	\centering
	\begin{tabular}{cccccc}
		\hline
		\textbf{Symbol} & $a$ & $v$ & $F_\text{b}\cdot 10^{-3}$ & $T$ & $\omega \cdot 10^{-3}$ \\ 
		\textbf{Bounds} & $[-3, 3]$ & $[2.2, 44.4]$ & $[0, 9]$ & $[15, 300]$ & $[0.9, 3]$ \\
		\hline
	\end{tabular}
	\caption{Variable bounds for the vehicle.}
	\label{tab:dynamics_coefficients}
\end{table}

For training the policy $\pi_\theta$ and for evaluation of the controllers, we consider episodic highway-driving scenarios with $\alpha = 0$.
Each episode requires a vehicle, initialized with a velocity in the range $[v_\text{min}+5, v_\text{max} -5]$ ms$^{-1}$, to track a random reference trajectory for 100s. 
Randomized reference trajectories are constructed as follows.
Beginning with velocity $x_{2, \text{ref}}(0) \sim \mathcal{U}(15, 25)$, the acceleration of the reference trajectory changes over five randomly spaced intervals.
For the first and last interval the acceleration is zero, with random values in $[-0.6, 0.6]$ ms$^{-2}$ for the other intervals.
Additionally, the reference velocity is clipped to the range $[5, 28]$ ms$^{-1}$ (18--100 kmh$^{-1}$).
To train $\pi_\theta$ with supervised learning the dataset $\mathcal{T}$ is generated using the MIQP-based MPC controller.
While the solution provided by MIQP is an approximation of the MINLP solution, we found the quality sufficient to train the policy $\pi_\theta$, and the computation time required to generate the data less.
Data is generated from 300 episodes, with $N=15$, and used to train an RNN with 4 layers of 256 features in the hidden state, followed by a fully connected linear layer.

\subsection{Evaluation}
To evaluate the performance of the controllers we compare the performance metric $P$, defined in \eqref{eq:metric}, over 100 episodes (not present in the training of $\pi_\theta$).
We select $\beta = 0.01$, tuned to balance the relative importance of the fuel consumption and the quadratic tracking error, and $Q=\text{diag}(1, 0.1)$.
All MPC controllers use a horizon of $N=15$, and both the MPC controllers and the underlying simulation use a time step of $\Delta t = 1$s.
\begin{color}{black}The backup solution is defined via $\phi(x_2) = \max_{j^\prime \in \Phi(x_2)}j^\prime$.\end{color}
For the first time step of each episode there are no shifted solutions available.
Hence, for the proposed approach the policy $\pi_\theta$ is not used for the fist time step, with instead the MIQP-based MPC problem providing $\textbf{j}$.

Using MINLP-based MPC (denoted NM) as a baseline, define the cost increase introduced by each controller as
\begin{equation}\label{eq:cost_increase}
	\Delta P_{\text{type}} = 100 \cdot \frac{P_\text{type} - P_\text{NM}}{P_\text{NM}},
\end{equation}
with $\text{type} \in \{\text{QM}, \text{LM}, \text{HM}\}$ representing the MIQP-based MPC, the proposed approach, and the hierarchical MPC, respectively.
Figure \ref{fig:eval} shows a box-and-whiskers plot of $\Delta P$ and the solve time required for each controller\footnote{Knitro experienced occasional numerical issues when solving \eqref{eq:MINLP}. In these cases, the MINLP solver Bonmin \cite{bonami2007bonmin} is used as a backup solver.} 
\begin{color}{black}It can be seen that LM requires significantly less computation time than QM and NM, with QM and NM unable to always find a solution within the MPC sample time $\Delta t$.\end{color}
Furthermore, the performance drop is negligible, with the median even improving over QM, likely due to the use of the exact fuel and friction models in the prediction model.
In contrast, while HM requires even less computation time than LM, as an even simpler NLP is solved, the performance drop is significant.
For LM the backup solution was used for $0.27\%$ of time steps.
For these time steps the average 1-norm difference between the policy and backup gear-shift schedules $|\pi_\theta(\bar{\textbf{x}}, \bar{\textbf{u}}^\prime, \textbf{x}_\text{ref}, \bar{\textbf{j}}) - \sigma(x_2)|_1$ was equal to $12$.
State and input trajectories for each controller on a representative episode are given in Figure \ref{fig:traj}.

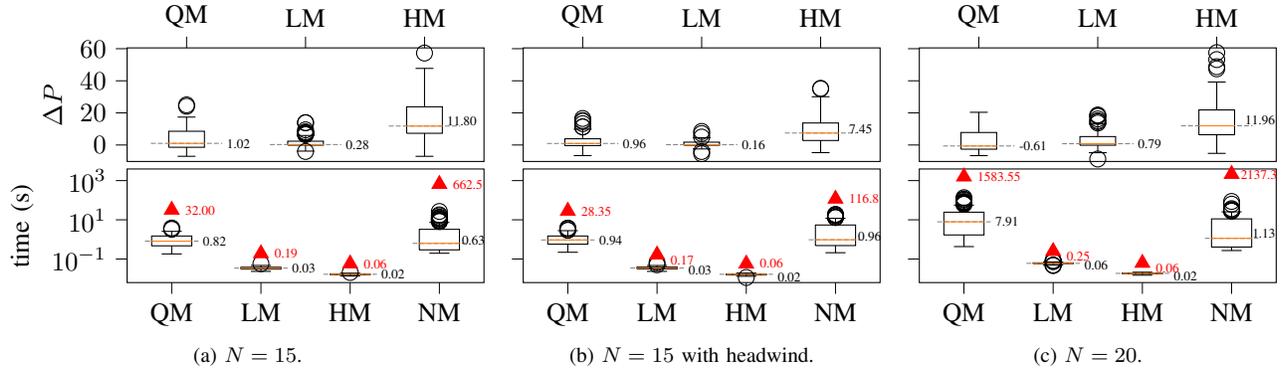
\begin{figure*}
	\centering
	\subfloat[$N=15$.\label{fig:eval}]{
\begin{tikzpicture}

\definecolor{darkgray176}{RGB}{176,176,176}
\definecolor{darkorange25512714}{RGB}{255,127,14}
\definecolor{gray}{RGB}{128,128,128}

\newcommand{\mW}{0.75*\axisdefaultwidth}
\newcommand{\mH}{0.8*\axisdefaultheight}

\begin{groupplot}[group style={group size=1 by 2, vertical sep = 0.1cm}]
\nextgroupplot[
height=\mH,
width=\mW,
scaled x ticks=manual:{}{\pgfmathparse{#1}},
tick align=outside,
tick pos=left,
tick scale binop=\times,
x grid style={darkgray176},
xmin=0.5, xmax=3.5,
xtick pos=right,
xtick style={color=black},
xtick={1,2,3},
xticklabels={QM,LM,HM},
y grid style={darkgray176},
ylabel={\(\Delta P\)},
ymin=-10.2637199798344, ymax=60.5012217079805,
ytick pos=left,
ytick style={color=black}
]
\addplot [black]
table {%
0.85 -1.41935778087897
1.15 -1.41935778087897
1.15 8.55642924709976
0.85 8.55642924709976
0.85 -1.41935778087897
};
\addplot [black]
table {%
1 -1.41935778087897
1 -7.03818660412987
};
\addplot [black]
table {%
1 8.55642924709976
1 17.3870949294238
};
\addplot [black]
table {%
0.925 -7.03818660412987
1.075 -7.03818660412987
};
\addplot [black]
table {%
0.925 17.3870949294238
1.075 17.3870949294238
};
\addplot [black, mark=o, mark size=3, mark options={solid,fill opacity=0}, only marks]
table {%
1 24.1262493980487
1 24.9727611398451
};
\addplot [black]
table {%
1.85 -0.112614731223458
2.15 -0.112614731223458
2.15 2.36378182710776
1.85 2.36378182710776
1.85 -0.112614731223458
};
\addplot [black]
table {%
2 -0.112614731223458
2 -3.80765403017423
};
\addplot [black]
table {%
2 2.36378182710776
2 5.91330541734733
};
\addplot [black]
table {%
1.925 -3.80765403017423
2.075 -3.80765403017423
};
\addplot [black]
table {%
1.925 5.91330541734733
2.075 5.91330541734733
};
\addplot [black, mark=o, mark size=3, mark options={solid,fill opacity=0}, only marks]
table {%
2 -4.02028027723297
2 8.19547069299594
2 6.45187684116487
2 6.93104316839888
2 13.7822608366038
2 7.55094863141199
2 7.56606652811417
2 6.67141584254597
2 9.43018875207574
2 13.9773103699074
2 7.39716190609549
2 7.36620247046713
};
\addplot [black]
table {%
2.85 7.29517893181112
3.15 7.29517893181112
3.15 23.792443473072
2.85 23.792443473072
2.85 7.29517893181112
};
\addplot [black]
table {%
3 7.29517893181112
3 -7.0471317212974
};
\addplot [black]
table {%
3 23.792443473072
3 47.7237674856573
};
\addplot [black]
table {%
2.925 -7.0471317212974
3.075 -7.0471317212974
};
\addplot [black]
table {%
2.925 47.7237674856573
3.075 47.7237674856573
};
\addplot [black, mark=o, mark size=3, mark options={solid,fill opacity=0}, only marks]
table {%
3 57.2846334494435
};
\path [draw=gray, very thin, dash pattern=on 1.85pt off 0.8pt]
(axis cs:0.7,1.0225913009127)
--(axis cs:1.3,1.0225913009127);

\path [draw=gray, very thin, dash pattern=on 1.85pt off 0.8pt]
(axis cs:1.7,0.283354630812836)
--(axis cs:2.3,0.283354630812836);

\path [draw=gray, very thin, dash pattern=on 1.85pt off 0.8pt]
(axis cs:2.7,11.7992364139636)
--(axis cs:3.3,11.7992364139636);

\addplot [darkorange25512714]
table {%
0.85 1.0225913009127
1.15 1.0225913009127
};
\addplot [darkorange25512714]
table {%
1.85 0.283354630812836
2.15 0.283354630812836
};
\addplot [darkorange25512714]
table {%
2.85 11.7992364139636
3.15 11.7992364139636
};
\draw (axis cs:1.3,1.0225913009127) node[
  scale=0.5,
  anchor=west,
  text=black,
  rotate=0.0
]{1.02};
\draw (axis cs:2.3,0.283354630812836) node[
  scale=0.5,
  anchor=west,
  text=black,
  rotate=0.0
]{0.28};
\draw (axis cs:3.15,15) node[
  scale=0.5,
  anchor=west,
  text=black,
  rotate=0.0
]{11.80};

\nextgroupplot[
height=\mH,
width=\mW,
ylabel=time (s),
log basis y={10},
tick align=outside,
tick pos=left,
tick scale binop=\times,
x grid style={darkgray176},
xmin=0.5, xmax=4.5,
xtick style={color=black},
xtick={1,2,3,4},
xticklabels={QM,LM,HM,NM},
y grid style={darkgray176},
ymin=0.00626415093484417, ymax=3920.57631598551,
ymode=log,
ytick style={color=black},
ytick={1e-05,0.001,0.1,10,1000,100000,10000000},
yticklabels={
	\(\displaystyle {10^{-5}}\),
	\(\displaystyle {10^{-3}}\),
	\(\displaystyle {10^{-1}}\),
	\(\displaystyle {10^{1}}\),
	\(\displaystyle {10^{3}}\),
	\(\displaystyle {10^{5}}\),
	\(\displaystyle {10^{7}}\)
}
]
\addplot [black]
table {%
0.775 0.4698463925
1.225 0.4698463925
1.225 1.4894727825
0.775 1.4894727825
0.775 0.4698463925
};
\addplot [black]
table {%
1 0.4698463925
1 0.18130622
};
\addplot [black]
table {%
1 1.4894727825
1 2.57537356
};
\addplot [black]
table {%
0.8875 0.18130622
1.1125 0.18130622
};
\addplot [black]
table {%
0.8875 2.57537356
1.1125 2.57537356
};
\addplot [black, mark=o, mark size=3, mark options={solid,fill opacity=0}, only marks]
table {%
1 3.79137132
1 3.57766409
1 3.2516627
1 3.5396824
};
\addplot [black]
table {%
1.775 0.030484695
2.225 0.030484695
2.225 0.0380090075
1.775 0.0380090075
1.775 0.030484695
};
\addplot [black]
table {%
2 0.030484695
2 0.02338722
};
\addplot [black]
table {%
2 0.0380090075
2 0.04692563
};
\addplot [black]
table {%
1.8875 0.02338722
2.1125 0.02338722
};
\addplot [black]
table {%
1.8875 0.04692563
2.1125 0.04692563
};
\addplot [black, mark=o, mark size=3, mark options={solid,fill opacity=0}, only marks]
table {%
2 0.05790369
};
\addplot [black]
table {%
2.775 0.015757725
3.225 0.015757725
3.225 0.017462945
2.775 0.017462945
2.775 0.015757725
};
\addplot [black]
table {%
3 0.015757725
3 0.01353462
};
\addplot [black]
table {%
3 0.017462945
3 0.0188649
};
\addplot [black]
table {%
2.8875 0.01353462
3.1125 0.01353462
};
\addplot [black]
table {%
2.8875 0.0188649
3.1125 0.0188649
};
\addplot [black, mark=o, mark size=3, mark options={solid,fill opacity=0}, only marks]
table {%
3 0.02061706
};
\addplot [black]
table {%
3.775 0.29289742
4.225 0.29289742
4.225 3.26988801
3.775 3.26988801
3.775 0.29289742
};
\addplot [black]
table {%
4 0.29289742
4 0.19683725
};
\addplot [black]
table {%
4 3.26988801
4 7.34310253
};
\addplot [black]
table {%
3.8875 0.19683725
4.1125 0.19683725
};
\addplot [black]
table {%
3.8875 7.34310253
4.1125 7.34310253
};
\addplot [black, mark=o, mark size=3, mark options={solid,fill opacity=0}, only marks]
table {%
4 8.19729381818182
4 12.51600074
4 26.58661901
4 16.22203247
4 8.42588299
4 8.82654386
4 11.1706538282828
4 10.97462942
4 13.75653293
4 9.68467502
4 17.91914342
};
\path [draw=gray, very thin, dash pattern=on 1.85pt off 0.8pt]
(axis cs:0.7,0.81643432)
--(axis cs:1.3,0.81643432);

\addplot [semithick, red, mark=triangle*, mark size=3, mark options={solid}]
table {%
1 31.996338
};
\path [draw=gray, very thin, dash pattern=on 1.85pt off 0.8pt]
(axis cs:1.7,0.03422805)
--(axis cs:2.3,0.03422805);

\addplot [semithick, red, mark=triangle*, mark size=3, mark options={solid}]
table {%
2 0.191005
};
\path [draw=gray, very thin, dash pattern=on 1.85pt off 0.8pt]
(axis cs:2.7,0.016633815)
--(axis cs:3.3,0.016633815);

\addplot [semithick, red, mark=triangle*, mark size=3, mark options={solid}]
table {%
3 0.058032
};
\path [draw=gray, very thin, dash pattern=on 1.85pt off 0.8pt]
(axis cs:3.7,0.6269573)
--(axis cs:4.3,0.6269573);

\addplot [semithick, red, mark=triangle*, mark size=3, mark options={solid}]
table {%
4 662.503133
};
\addplot [darkorange25512714]
table {%
0.775 0.81643432
1.225 0.81643432
};
\addplot [darkorange25512714]
table {%
1.775 0.03422805
2.225 0.03422805
};
\addplot [darkorange25512714]
table {%
2.775 0.016633815
3.225 0.016633815
};
\addplot [darkorange25512714]
table {%
3.775 0.6269573
4.225 0.6269573
};
\draw (axis cs:1.3,0.81643432) node[
  scale=0.5,
  anchor=west,
  text=black,
  rotate=0.0
]{0.82};
\draw (axis cs:1.1,31.996338) node[
  scale=0.5,
  anchor=west,
  text=red,
  rotate=0.0
]{32.00};
\draw (axis cs:2.3,0.03422805) node[
  scale=0.5,
  anchor=west,
  text=black,
  rotate=0.0
]{0.03};
\draw (axis cs:2.1,0.191005) node[
  scale=0.5,
  anchor=west,
  text=red,
  rotate=0.0
]{0.19};
\draw (axis cs:3.3,0.016633815) node[
  scale=0.5,
  anchor=west,
  text=black,
  rotate=0.0
]{0.02};
\draw (axis cs:3.1,0.058032) node[
  scale=0.5,
  anchor=west,
  text=red,
  rotate=0.0
]{0.06};
\draw (axis cs:4.2,0.95) node[
  scale=0.5,
  anchor=west,
  text=black,
  rotate=0.0
]{0.63};
\draw (axis cs:4.1,662.5) node[
  scale=0.5,
  anchor=west,
  text=red,
  rotate=0.0
]{662.5};
\end{groupplot}

\end{tikzpicture}}
	\subfloat[$N=15$ with headwind.\label{fig:eval_wind}]{
\begin{tikzpicture}

\definecolor{darkgray176}{RGB}{176,176,176}
\definecolor{darkorange25512714}{RGB}{255,127,14}
\definecolor{gray}{RGB}{128,128,128}

\newcommand{\mW}{0.75*\axisdefaultwidth}
\newcommand{\mH}{0.8*\axisdefaultheight}

\begin{groupplot}[group style={group size=1 by 2, vertical sep = 0.1cm}]
\nextgroupplot[
height=\mH,
width=\mW,
scaled x ticks=manual:{}{\pgfmathparse{#1}},
tick align=outside,
tick pos=left,
tick scale binop=\times,
x grid style={darkgray176},
xmin=0.5, xmax=3.5,
xtick pos=right,
xtick style={color=black},
xtick={1,2,3},
xticklabels={QM,LM,HM},
y grid style={darkgray176},
ymin=-10.2637199798344, ymax=60.5012217079805,
ytick pos=left,
ytick style={color=black},
yticklabels={}
]
\addplot [black]
table {%
0.85 -0.373761653964891
1.15 -0.373761653964891
1.15 3.87669010675658
0.85 3.87669010675658
0.85 -0.373761653964891
};
\addplot [black]
table {%
1 -0.373761653964891
1 -6.60319951118018
};
\addplot [black]
table {%
1 3.87669010675658
1 9.83964875153002
};
\addplot [black]
table {%
0.925 -6.60319951118018
1.075 -6.60319951118018
};
\addplot [black]
table {%
0.925 9.83964875153002
1.075 9.83964875153002
};
\addplot [black, mark=o, mark size=3, mark options={solid,fill opacity=0}, only marks]
table {%
1 11.2963955584055
1 13.3049313518276
1 15.224624892402
1 11.2691514479559
1 14.8281811455399
1 16.5105576199999
};
\addplot [black]
table {%
1.85 -0.139280317366536
2.15 -0.139280317366536
2.15 1.76218345729149
1.85 1.76218345729149
1.85 -0.139280317366536
};
\addplot [black]
table {%
2 -0.139280317366536
2 -2.46453250811353
};
\addplot [black]
table {%
2 1.76218345729149
2 4.57898851723963
};
\addplot [black]
table {%
1.925 -2.46453250811353
2.075 -2.46453250811353
};
\addplot [black]
table {%
1.925 4.57898851723963
2.075 4.57898851723963
};
\addplot [black, mark=o, mark size=3, mark options={solid,fill opacity=0}, only marks]
table {%
2 -4.06599719091479
2 -5.42931131923265
2 4.94143801657543
2 8.46576526487015
2 7.14098862181723
};
\addplot [black]
table {%
2.85 2.80457685409422
3.15 2.80457685409422
3.15 13.7162840858868
2.85 13.7162840858868
2.85 2.80457685409422
};
\addplot [black]
table {%
3 2.80457685409422
3 -4.75319212919879
};
\addplot [black]
table {%
3 13.7162840858868
3 30.0708199332805
};
\addplot [black]
table {%
2.925 -4.75319212919879
3.075 -4.75319212919879
};
\addplot [black]
table {%
2.925 30.0708199332805
3.075 30.0708199332805
};
\addplot [black, mark=o, mark size=3, mark options={solid,fill opacity=0}, only marks]
table {%
3 35.0308355714046
3 35.2908896470761
};
\path [draw=gray, very thin, dash pattern=on 1.85pt off 0.8pt]
(axis cs:0.7,0.961720580091415)
--(axis cs:1.3,0.961720580091415);

\path [draw=gray, very thin, dash pattern=on 1.85pt off 0.8pt]
(axis cs:1.7,0.158818417985345)
--(axis cs:2.3,0.158818417985345);

\path [draw=gray, very thin, dash pattern=on 1.85pt off 0.8pt]
(axis cs:2.7,7.45316020732673)
--(axis cs:3.3,7.45316020732673);

\addplot [darkorange25512714]
table {%
0.85 0.961720580091415
1.15 0.961720580091415
};
\addplot [darkorange25512714]
table {%
1.85 0.158818417985345
2.15 0.158818417985345
};
\addplot [darkorange25512714]
table {%
2.85 7.45316020732673
3.15 7.45316020732673
};
\draw (axis cs:1.3,0.961720580091415) node[
  scale=0.5,
  anchor=west,
  text=black,
  rotate=0.0
]{0.96};
\draw (axis cs:2.3,0.13863723244363) node[
  scale=0.5,
  anchor=west,
  text=black,
  rotate=0.0
]{0.16};
\draw (axis cs:3.2,10) node[
  scale=0.5,
  anchor=west,
  text=black,
  rotate=0.0
]{7.45};

\nextgroupplot[
height=\mH,
width=\mW,
log basis y={10},
tick align=outside,
tick pos=left,
tick scale binop=\times,
x grid style={darkgray176},
xmin=0.5, xmax=4.5,
xtick style={color=black},
xtick={1,2,3,4},
xticklabels={QM,LM,HM,NM},
y grid style={darkgray176},
ymin=0.00626415093484417, ymax=3920.57631598551,
ymode=log,
ytick style={color=black},
ytick={1e-05,0.001,0.1,10,1000,100000,10000000},
yticklabels={}
]
\addplot [black]
table {%
0.775 0.5714569825
1.225 0.5714569825
1.225 1.4919861275
0.775 1.4919861275
0.775 0.5714569825
};
\addplot [black]
table {%
1 0.5714569825
1 0.22071672
};
\addplot [black]
table {%
1 1.4919861275
1 2.79295687
};
\addplot [black]
table {%
0.8875 0.22071672
1.1125 0.22071672
};
\addplot [black]
table {%
0.8875 2.79295687
1.1125 2.79295687
};
\addplot [black, mark=o, mark size=3, mark options={solid,fill opacity=0}, only marks]
table {%
1 3.74221861
1 3.02943756
1 2.90379185
1 3.49136958
1 3.72801159
};
\addplot [black]
table {%
1.775 0.0313386875
2.225 0.0313386875
2.225 0.0384791775
1.775 0.0384791775
1.775 0.0313386875
};
\addplot [black]
table {%
2 0.0313386875
2 0.02346793
};
\addplot [black]
table {%
2 0.0384791775
2 0.04662599
};
\addplot [black]
table {%
1.8875 0.02346793
2.1125 0.02346793
};
\addplot [black]
table {%
1.8875 0.04662599
2.1125 0.04662599
};
\addplot [black, mark=o, mark size=3, mark options={solid,fill opacity=0}, only marks]
table {%
2 0.06493177
2 0.05083087
};
\addplot [black]
table {%
2.775 0.015280755
3.225 0.015280755
3.225 0.017392325
2.775 0.017392325
2.775 0.015280755
};
\addplot [black]
table {%
3 0.015280755
3 0.01312598
};
\addplot [black]
table {%
3 0.017392325
3 0.01908473
};
\addplot [black]
table {%
2.8875 0.01312598
3.1125 0.01312598
};
\addplot [black]
table {%
2.8875 0.01908473
3.1125 0.01908473
};
\addplot [black, mark=o, mark size=3, mark options={solid,fill opacity=0}, only marks]
table {%
3 0.01137967
};
\addplot [black]
table {%
3.775 0.48492912
4.225 0.48492912
4.225 5.429599185
3.775 5.429599185
3.775 0.48492912
};
\addplot [black]
table {%
4 0.48492912
4 0.20893914
};
\addplot [black]
table {%
4 5.429599185
4 11.92932359
};
\addplot [black]
table {%
3.8875 0.20893914
4.1125 0.20893914
};
\addplot [black]
table {%
3.8875 11.92932359
4.1125 11.92932359
};
\addplot [black, mark=o, mark size=3, mark options={solid,fill opacity=0}, only marks]
table {%
4 13.13205711
4 15.96387688
4 16.52874718
4 16.40358513
4 13.31877614
4 18.83983166
4 18.4452843
4 13.35687524
};
\path [draw=gray, very thin, dash pattern=on 1.85pt off 0.8pt]
(axis cs:0.7,0.93892199)
--(axis cs:1.3,0.93892199);

\addplot [semithick, red, mark=triangle*, mark size=3, mark options={solid}]
table {%
1 28.346556
};
\path [draw=gray, very thin, dash pattern=on 1.85pt off 0.8pt]
(axis cs:1.7,0.03457225)
--(axis cs:2.3,0.03457225);

\addplot [semithick, red, mark=triangle*, mark size=3, mark options={solid}]
table {%
2 0.165471
};
\path [draw=gray, very thin, dash pattern=on 1.85pt off 0.8pt]
(axis cs:2.7,0.016566695)
--(axis cs:3.3,0.016566695);

\addplot [semithick, red, mark=triangle*, mark size=3, mark options={solid}]
table {%
3 0.057121
};
\path [draw=gray, very thin, dash pattern=on 1.85pt off 0.8pt]
(axis cs:3.7,0.959167555)
--(axis cs:4.3,0.959167555);

\addplot [semithick, red, mark=triangle*, mark size=3, mark options={solid}]
table {%
4 116.792592
};
\addplot [darkorange25512714]
table {%
0.775 0.93892199
1.225 0.93892199
};
\addplot [darkorange25512714]
table {%
1.775 0.03457225
2.225 0.03457225
};
\addplot [darkorange25512714]
table {%
2.775 0.016566695
3.225 0.016566695
};
\addplot [darkorange25512714]
table {%
3.775 0.959167555
4.225 0.959167555
};
\draw (axis cs:1.3,0.93892199) node[
  scale=0.5,
  anchor=west,
  text=black,
  rotate=0.0
]{0.94};
\draw (axis cs:1.1,28.346556) node[
  scale=0.5,
  anchor=west,
  text=red,
  rotate=0.0
]{28.35};
\draw (axis cs:2.3,0.0272039) node[
  scale=0.5,
  anchor=west,
  text=black,
  rotate=0.0
]{0.03};
\draw (axis cs:2.1,0.092295) node[
  scale=0.5,
  anchor=west,
  text=red,
  rotate=0.0
]{0.17};
\draw (axis cs:3.3,0.01212531) node[
  scale=0.5,
  anchor=west,
  text=black,
  rotate=0.0
]{0.02};
\draw (axis cs:3.1,0.0626438) node[
  scale=0.5,
  anchor=west,
  text=red,
  rotate=0.0
]{0.06};
\draw (axis cs:4.2,1.5) node[
  scale=0.5,
  anchor=west,
  text=black,
  rotate=0.0
]{0.96};
\draw (axis cs:4.1,116.792592) node[
  scale=0.5,
  anchor=west,
  text=red,
  rotate=0.0
]{116.8};
\end{groupplot}

\end{tikzpicture}}
	\subfloat[$N=20$.\label{fig:eval_N_20}]{
\begin{tikzpicture}

\definecolor{darkgray176}{RGB}{176,176,176}
\definecolor{darkorange25512714}{RGB}{255,127,14}
\definecolor{gray}{RGB}{128,128,128}

\newcommand{\mW}{0.75*\axisdefaultwidth}
\newcommand{\mH}{0.8*\axisdefaultheight}

\begin{groupplot}[group style={group size=1 by 2, vertical sep = 0.1cm}]
\nextgroupplot[
height=\mH,
width=\mW,
scaled x ticks=manual:{}{\pgfmathparse{#1}},
tick align=outside,
tick pos=left,
tick scale binop=\times,
x grid style={darkgray176},
xmin=0.5, xmax=3.5,
xtick pos=right,
xtick style={color=black},
xtick={1,2,3},
xticklabels={QM,LM,HM},
y grid style={darkgray176},
ymin=-10.2637199798344, ymax=60.5012217079805,
ytick pos=left,
ytick style={color=black},
yticklabels={}
]
\addplot [black]
table {%
0.85 -2.60781873330854
1.15 -2.60781873330854
1.15 7.79056981199698
0.85 7.79056981199698
0.85 -2.60781873330854
};
\addplot [black]
table {%
1 -2.60781873330854
1 -6.62338264966304
};
\addplot [black]
table {%
1 7.79056981199698
1 20.3792069748875
};
\addplot [black]
table {%
0.925 -6.62338264966304
1.075 -6.62338264966304
};
\addplot [black]
table {%
0.925 20.3792069748875
1.075 20.3792069748875
};
\addplot [black]
table {%
1.85 -0.144020818135947
2.15 -0.144020818135947
2.15 5.24268097253421
1.85 5.24268097253421
1.85 -0.144020818135947
};
\addplot [black]
table {%
2 -0.144020818135947
2 -4.79384337269877
};
\addplot [black]
table {%
2 5.24268097253421
2 13.0690379443284
};
\addplot [black]
table {%
1.925 -4.79384337269877
2.075 -4.79384337269877
};
\addplot [black]
table {%
1.925 13.0690379443284
2.075 13.0690379443284
};
\addplot [black, mark=o, mark size=3, mark options={solid,fill opacity=0}, only marks]
table {%
2 -8.73145116503373
2 18.3627693596937
2 18.6783854693965
2 17.8749925228041
2 16.2273668782044
2 14.5142226130713
2 15.0654237424834
2 13.6724032050701
};
\addplot [black]
table {%
2.85 6.39825931391753
3.15 6.39825931391753
3.15 21.8529558088328
2.85 21.8529558088328
2.85 6.39825931391753
};
\addplot [black]
table {%
3 6.39825931391753
3 -5.21082900798735
};
\addplot [black]
table {%
3 21.8529558088328
3 39.2315021860511
};
\addplot [black]
table {%
2.925 -5.21082900798735
3.075 -5.21082900798735
};
\addplot [black]
table {%
2.925 39.2315021860511
3.075 39.2315021860511
};
\addplot [black, mark=o, mark size=3, mark options={solid,fill opacity=0}, only marks]
table {%
3 48.7884309391171
3 57.5497021905681
3 53.2114834022246
3 47.4576532723387
};
\path [draw=gray, very thin, dash pattern=on 1.85pt off 0.8pt]
(axis cs:0.7,-0.6126847831035)
--(axis cs:1.3,-0.6126847831035);

\path [draw=gray, very thin, dash pattern=on 1.85pt off 0.8pt]
(axis cs:1.7,0.789998020606323)
--(axis cs:2.3,0.789998020606323);

\path [draw=gray, very thin, dash pattern=on 1.85pt off 0.8pt]
(axis cs:2.7,11.9587872832521)
--(axis cs:3.3,11.9587872832521);

\addplot [darkorange25512714]
table {%
0.85 -0.6126847831035
1.15 -0.6126847831035
};
\addplot [darkorange25512714]
table {%
1.85 0.789998020606323
2.15 0.789998020606323
};
\addplot [darkorange25512714]
table {%
2.85 11.9587872832521
3.15 11.9587872832521
};
\draw (axis cs:1.3,-0.6126847831035) node[
  scale=0.5,
  anchor=west,
  text=black,
  rotate=0.0
]{-0.61};
\draw (axis cs:2.3,0.695202266059774) node[
  scale=0.5,
  anchor=west,
  text=black,
  rotate=0.0
]{0.79};
\draw (axis cs:3.2,15.5) node[
  scale=0.5,
  anchor=west,
  text=black,
  rotate=0.0
]{11.96};

\nextgroupplot[
height=\mH,
width=\mW,
log basis y={10},
tick align=outside,
tick pos=left,
tick scale binop=\times,
x grid style={darkgray176},
xmin=0.5, xmax=4.5,
xtick style={color=black},
xtick={1,2,3,4},
xticklabels={QM,LM,HM,NM},
y grid style={darkgray176},
ymin=0.00626415093484417, ymax=3920.57631598551,
ymode=log,
ytick style={color=black},
ytick={1e-05,0.001,0.1,10,1000,100000,10000000},
yticklabels={}
]
\addplot [black]
table {%
0.775 1.6862862
1.225 1.6862862
1.225 24.40866503
0.775 24.40866503
0.775 1.6862862
};
\addplot [black]
table {%
1 1.6862862
1 0.42988879
};
\addplot [black]
table {%
1 24.40866503
1 56.24662122
};
\addplot [black]
table {%
0.8875 0.42988879
1.1125 0.42988879
};
\addplot [black]
table {%
0.8875 56.24662122
1.1125 56.24662122
};
\addplot [black, mark=o, mark size=3, mark options={solid,fill opacity=0}, only marks]
table {%
1 64.24503184
1 134.82635467
1 67.47035566
1 75.62975956
1 99.97901681
1 80.71097514
1 113.10856944
1 61.5700859
1 85.7929367100001
1 78.58151167
1 75.69655017
1 114.33687444
1 73.67998395
};
\addplot [black]
table {%
1.775 0.057494465
2.225 0.057494465
2.225 0.0636180225
1.775 0.0636180225
1.775 0.057494465
};
\addplot [black]
table {%
2 0.057494465
2 0.04843763
};
\addplot [black]
table {%
2 0.0636180225
2 0.07040062
};
\addplot [black]
table {%
1.8875 0.04843763
2.1125 0.04843763
};
\addplot [black]
table {%
1.8875 0.07040062
2.1125 0.07040062
};
\addplot [black, mark=o, mark size=3, mark options={solid,fill opacity=0}, only marks]
table {%
2 0.04820865
2 0.04773399
2 0.04713834
2 0.0774587
2 0.08357099
2 0.07901262
2 0.07357196
};
\addplot [black]
table {%
2.775 0.0173388225
3.225 0.0173388225
3.225 0.019509995
2.775 0.019509995
2.775 0.0173388225
};
\addplot [black]
table {%
3 0.0173388225
3 0.01493653
};
\addplot [black]
table {%
3 0.019509995
3 0.02110971
};
\addplot [black]
table {%
2.8875 0.01493653
3.1125 0.01493653
};
\addplot [black]
table {%
2.8875 0.02110971
3.1125 0.02110971
};
\addplot [black]
table {%
3.775 0.407343805
4.225 0.407343805
4.225 11.1332591544697
3.775 11.1332591544697
3.775 0.407343805
};
\addplot [black]
table {%
4 0.407343805
4 0.26982801
};
\addplot [black]
table {%
4 11.1332591544697
4 25.47839431
};
\addplot [black]
table {%
3.8875 0.26982801
4.1125 0.26982801
};
\addplot [black]
table {%
3.8875 25.47839431
4.1125 25.47839431
};
\addplot [black, mark=o, mark size=3, mark options={solid,fill opacity=0}, only marks]
table {%
4 28.92314804
4 27.81122097
4 86.47236291
4 62.8193652
4 29.23185254
4 36.97507966
4 33.3462661212121
4 29.51604055
};
\path [draw=gray, very thin, dash pattern=on 1.85pt off 0.8pt]
(axis cs:0.7,7.90738767)
--(axis cs:1.3,7.90738767);

\addplot [semithick, red, mark=triangle*, mark size=3, mark options={solid}]
table {%
1 1583.5521
};
\path [draw=gray, very thin, dash pattern=on 1.85pt off 0.8pt]
(axis cs:1.7,0.06008081)
--(axis cs:2.3,0.06008081);

\addplot [semithick, red, mark=triangle*, mark size=3, mark options={solid}]
table {%
2 0.250043
};
\path [draw=gray, very thin, dash pattern=on 1.85pt off 0.8pt]
(axis cs:2.7,0.01841381)
--(axis cs:3.3,0.01841381);

\addplot [semithick, red, mark=triangle*, mark size=3, mark options={solid}]
table {%
3 0.062423
};
\path [draw=gray, very thin, dash pattern=on 1.85pt off 0.8pt]
(axis cs:3.7,1.133737135)
--(axis cs:4.3,1.133737135);

\addplot [semithick, red, mark=triangle*, mark size=3, mark options={solid}]
table {%
4 2137.337957
};
\addplot [darkorange25512714]
table {%
0.775 7.90738767
1.225 7.90738767
};
\addplot [darkorange25512714]
table {%
1.775 0.06008081
2.225 0.06008081
};
\addplot [darkorange25512714]
table {%
2.775 0.01841381
3.225 0.01841381
};
\addplot [darkorange25512714]
table {%
3.775 1.133737135
4.225 1.133737135
};
\draw (axis cs:1.3,7.90738767) node[
  scale=0.5,
  anchor=west,
  text=black,
  rotate=0.0
]{7.91};
\draw (axis cs:1.1,1583.5521) node[
  scale=0.5,
  anchor=west,
  text=red,
  rotate=0.0
]{1583.55};
\draw (axis cs:2.3,0.049669015) node[
  scale=0.5,
  anchor=west,
  text=black,
  rotate=0.0
]{0.06};
\draw (axis cs:2.1,0.13696) node[
  scale=0.5,
  anchor=west,
  text=red,
  rotate=0.0
]{0.25};
\draw (axis cs:3.3,0.013932125) node[
  scale=0.5,
  anchor=west,
  text=black,
  rotate=0.0
]{0.02};
\draw (axis cs:3.1,0.035639) node[
  scale=0.5,
  anchor=west,
  text=red,
  rotate=0.0
]{0.06};
\draw (axis cs:4.18,1.95) node[
  scale=0.5,
  anchor=west,
  text=black,
  rotate=0.0
]{1.13};
\draw (axis cs:4.05,2137.337957) node[
  scale=0.5,
  anchor=west,
  text=red,
  rotate=0.0
]{2137.34};
\end{groupplot}

\end{tikzpicture}}
	\caption{
			Top: distribution of \eqref{eq:cost_increase} across 100 episodes with median marked. Bottom: distribution of average MPC solve time for each episode across 100 episodes with median marked. The maximal time of all steps is marked with a red triangle.}
	\label{fig:eval_all}
\end{figure*}
\begin{figure}
	\centering
	\vspace*{-0.75cm}
	\hspace{1.5cm}\subfloat{
\begin{tikzpicture}

\newcommand{\myW}{0.6*\axisdefaultwidth}
\newcommand{\myH}{0.8*\axisdefaultheight}

\definecolor{darkgray178}{RGB}{178,178,178}
\definecolor{darkolivegreen7012033}{RGB}{70,120,33}
\definecolor{firebrick166640}{RGB}{166,6,40}
\definecolor{lightgray204}{RGB}{204,204,204}
\definecolor{silver188}{RGB}{188,188,188}
\definecolor{slategray122104166}{RGB}{122,104,166}
\definecolor{steelblue52138189}{RGB}{52,138,189}
\definecolor{whitesmoke238}{RGB}{238,238,238}

\begin{axis}[
width=\myW,
height=\myH,
axis background/.style={fill=none}, 
axis line style={silver188},
hide x axis,
hide y axis,
legend cell align={left},
legend columns=4,
legend style={ 
  fill opacity=0.8,
  draw opacity=1,
  text opacity=1,
  at={(0.5,0.5)},
  anchor=center,
  draw=none,
  fill=none
},
tick pos=left,
x grid style={darkgray178},
xmajorgrids,
xmin=100, xmax=101,
xtick style={color=black},
xtick={4.6,4.8,5,5.2,5.4},
xticklabels={
  \(\displaystyle {4.6}\),
  \(\displaystyle {4.8}\),
  \(\displaystyle {5.0}\),
  \(\displaystyle {5.2}\),
  \(\displaystyle {5.4}\)
},
y grid style={darkgray178},
ymajorgrids,
ymin=4.725, ymax=5.275,
ytick style={color=black},
ytick={4.5,4.75,5,5.25,5.5},
yticklabels={
  \(\displaystyle {4.50}\),
  \(\displaystyle {4.75}\),
  \(\displaystyle {5.00}\),
  \(\displaystyle {5.25}\),
  \(\displaystyle {5.50}\)
}
]
\addplot [semithick, red, dash pattern=on 5.55pt off 2.4pt]
table {%
5 5
};
\addlegendentry{QM}
\addplot [semithick, blue]
table {%
5 5
};
\addlegendentry{LM}
\addplot [semithick, green, dash pattern=on 9.6pt off 2.4pt on 1.5pt off 2.4pt]
table {%
5 5
};
\addlegendentry{HM}
\addplot [semithick, orange, dash pattern=on 1.5pt off 2.475pt]
table {%
5 5
};
\addlegendentry{NM}

\end{axis}
\end{tikzpicture}}\\
	\vspace*{-1cm}
	\subfloat{\input{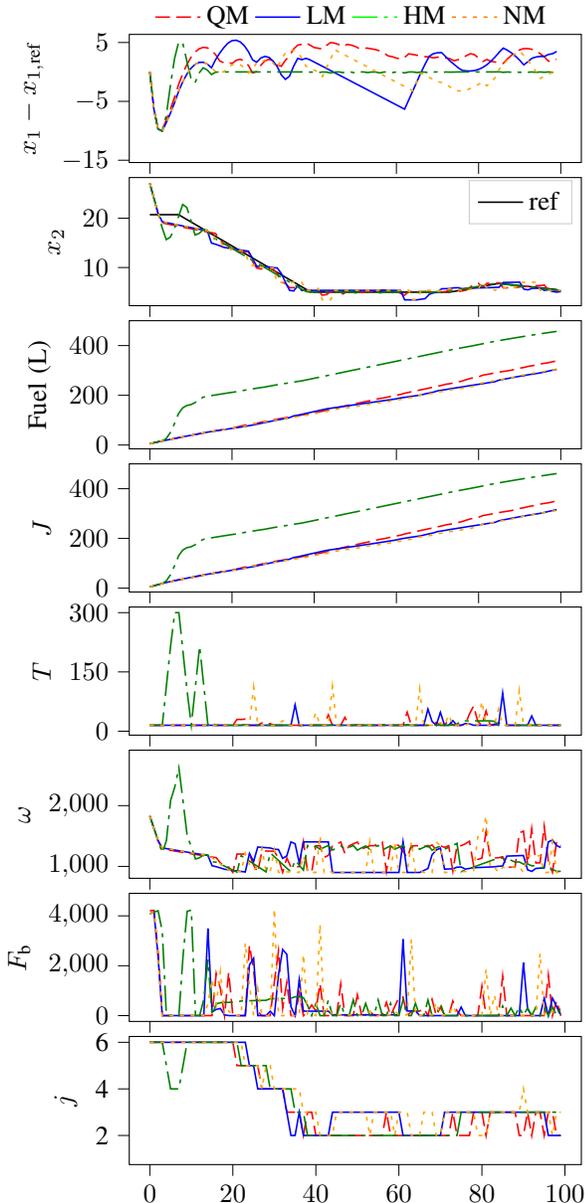}}\\
	\vspace*{-0.25cm}
	\caption{Representative trajectories for each controller.}
	\label{fig:traj}
\end{figure}

It is well known that MPC controllers have a level of inherent robustness thanks to re-optimization at each time step \cite{rawlingsModelPredictiveControl2017}.
To explore the preservation of this robustness in the proposed learning-based MPC controller, a further evaluation is conducted for 100 episodes where a strong disturbance in the form of a time-varying headwind $v_\text{w}(t) \in [8, 14]$ ms$^{-1}$ is present. 
This disturbance is unmodeled for the MPC controllers, but affects the true dynamics of the system.
The headwind changes the relative velocity of the vehicle when calculating wind drag, i.e., the drag term in \eqref{eq:dynamics} becomes $C (v + v_\text{w})^2$.
Figure \ref{fig:eval_wind} shows the cost increase and the solve time over 100 episodes under headwind disturbance.
LM retains a performance that is comparable to that of the mixed-integer approaches, with a superior computational burden.
The backup solution was used for $2.15\%$ of time steps.

Finally, we explore the scaling of the approaches with the prediction horizon $N$, and the generalization of the proposed approach to different horizon lengths.
An evaluation is conducted for 100 episodes (without headwind) with $N = 20$. 
The proposed approach LM uses the policy $\pi_\theta$ trained with $N=15$, i.e., no retraining. 
Figure \ref{fig:eval_N_20} shows the cost increase and the solve times.
Again, LM retains a comparable performance to the mixed-integer controllers with superior computation time, demonstrating how the learned policy can generalize to horizons longer than that on which it was trained.  
For LM the backup solution was used for $0.43\%$ of time steps.

\section{CONCLUSIONS}\label{sec:conclusions}
In this work we have proposed a novel learning-based MPC controller for fuel efficient autonomous driving.
By learning a policy that selects the gear-shift schedule over the MPC prediction horizon, the benefits of speed and gear co-optimization, i.e., fuel efficient tracking, are retained without the computational burden of solving a mixed-integer program.
The result is a controller that achieves a performance comparable to approaches that solve mixed-integer programs, and that has a computational burden comparable to sub-optimal approaches that decouple speed and gear optimization.
Future work will look at extending the approach to vehicle platoons and addressing model mismatch with robust MPC.

\appendices
\section{Proof of Proposition \ref{prop:1}}\label{proof:prop_1}
\begin{proof}
We prove Proposition \ref{prop:1} by showing the existence of a specific feasible solution to \eqref{eq:NLP}, namely a constant-velocity trajectory.
We show that $J\big(x(k), \textbf{x}_\mathrm{ref}(k), \textbf{j}(k)\big) < \infty$ with $\textbf{j}(k) = \sigma\big(x_2(k)\big)$ for the solution $\textbf{x}(k), \textbf{u}^\prime(k)$ with 
\begin{equation}\label{eq:state_sol}
	x(i|k) = \big[x_1(k) +  \tau x_2(k)\Delta t,  x_2(k)\big]^\top, \: i = 0,\dots,N
\end{equation}
and constant input $u^\prime(i|k) = [u_1, u_2]^\top$ for $i=0,\dots,N-1$, with $T_\text{min} \leq u_1 \leq T_\text{max}$, $ F_\text{b, min} \leq u_2 \leq F_\text{b, max}$ such that
\begin{equation}\label{eq:assumpion_in_prop}
	\begin{aligned}
		\frac{u_1 z\Big(\phi\big(x_2(k)\big)\Big) z_\text{f}}{r} - C x_2^2(k) - u_2 - G = 0.
	\end{aligned}
\end{equation}
By the assumption in Proposition \ref{prop:1} this control input exists for any $x_2(k)$ with $v_\mathrm{min} \leq x_2(k) \leq v_\mathrm{max}$ and for gear $j=\phi\big(x_2(k)\big)$.
The dynamics constraints \eqref{eq:nlp_dynam_cnstr} are satisfied by this solution as \eqref{eq:state_sol} and \eqref{eq:assumpion_in_prop} satisfy the equality constraint, i.e., for $i = 0,\dots,N-1$ we have
\begin{equation}
	f\Big(x(i|k), \big[u_1, u_2, \phi\big(x_2(k)\big)\big]^\top\Big) = x(i+1|k).
\end{equation}
Furthermore, the constraints \eqref{eq:nlp_rate_cnstr} and \eqref{eq:nlp_gear_shift_lim} are trivially satisfied by this solution with constant velocity $x_2(k)$, engine torque $u_1$, and gear $\phi\big(x_2(k)\big)$.
Finally, constraint \eqref{eq:nlp_state_input_cnstr} is satisfied as $T_\text{min} \leq u_1 \leq T_\text{max}$, $F_\text{b, max} \leq u_2 \leq F_\text{b, max}$, and $w_\text{min} \leq \omega\Big(x_2(k), \phi\big(x_2(k)\big)\Big) \leq w_\text{max}$ by the definition of $\phi$.
\begin{color}{black}Hence, all constraints are satisfied, with $\textbf{x}(k), \textbf{u}^\prime(k)$ one of potentially many feasible solutions.
Indeed, this solution could be used as a feasible initial guess for a numerical solver.\end{color}
\end{proof}

\bibliographystyle{plain}
\bibliography{root.bib}

\end{document}